\documentclass[sigconf]{acmart}

\usepackage{booktabs} % For formal tables
\usepackage[utf8]{inputenc}
\usepackage{mathtools}
\usepackage{tabularx}
\usepackage{changepage}
\usepackage{enumitem}
\usepackage{environ}
\usepackage{balance}

\usepackage{inputenc} 
\usepackage{mathtools}

\newtheorem{remark}{Remark}
\let\set\mathbb
\newcommand{\n}{z}
\newcommand{\p}{\zeta}
\newcommand{\T}{\phi}
\newcommand{\IT}{\phi^{-1}}

\newcommand{\N}{{\mathbb N}}

\newcommand{\MC}{{\mathbb C}}
\newcommand{\F}{{\mathcal C}}

\newcommand {\m }{ {\,{\bf mod}\,}}

\newcommand {\zm}{{\rm O}}
\newcommand {\id}{{\rm I}}
\newcommand {\ra}{{\rm rank}}
\newcommand {\ord}{{\rm ord}}

\newcommand {\rank}{{\rm rank}}

\newcommand {\lc}{{\rm lc}}

\newcommand {\den}{{\rm den}}
\newcommand {\num}{{\rm num}}
\newcommand{\SF}{{\mathcal O} }
\newcommand{\IA}{{A^*} }

\newlength\labwd
\makeatletter

\newcounter{enumv}[enumiv]
\newcounter{enumvi}[enumv]
\newenvironment{myenum}
    {\advance\@enumdepth\@ne
    \ifnum \@enumdepth >6\@toodeep\else
    \edef\@enumctr{enum\romannumeral\the\@enumdepth}

    \renewcommand\theenumii{\arabic{enumii}}
    \renewcommand\theenumiii{\arabic{enumiii}}
    \renewcommand\theenumiv{\arabic{enumiv}}
    \renewcommand\p@enumii{\theenumi.}
    \renewcommand\p@enumiii{\theenumi.\theenumii.}
    \renewcommand\p@enumiv{\theenumi.\theenumii.\theenumiii.}
    \renewcommand\p@enumv{\theenumi.\theenumii.\theenumiii.\theenumiv.}
    \renewcommand\p@enumvi{\theenumi.\theenumii.\theenumiii.\theenumiv.\theenumv.}
    \begin{list}{\csname label\@enumctr\endcsname}{%
        \usecounter\@enumctr
        \setlength\labelwidth{5pt}
        \setlength\labelsep{4pt}
        \setlength\leftmargin{15pt}
        \setlength\labwd{\ifcase\@enumdepth\or -15pt\or -30pt\or -45pt\or -60pt\or -75pt\or -90pt\fi}%
    }%
    \fi%
    }
    {\ifnum \@enumdepth >6\else\end{list}\fi}
\makeatother

\newcounter{algorithm}
\NewEnviron{algo}[3]
{\refstepcounter{algorithm}
  \begin{tabularx}{0.95\linewidth}{X}
      \hline\\[-2.5ex]
  \textbf{Algorithm \thealgorithm: } #1\\
  \hline\\[-2.5ex]
  {\begin{tabularx}{\linewidth}{l X}
      \textbf{Input:} & #2\\
      \textbf{Output:} & #3
    \end{tabularx}}\\
   \hline\\[-4ex]
   \begin{adjustwidth}{0.5cm}{0cm}
     \begin{myenum}\ignorespaces
       \BODY
     \end{myenum}
   \end{adjustwidth}\\[-3ex]
   \hline
 \end{tabularx}
}

\NewEnviron{while}[1]
{\textbf{WHILE} (#1) \textbf{DO}
  \begin{myenum}\ignorespaces
    \BODY
  \end{myenum}
}

\NewEnviron{algif}[1]
{\textbf{IF} (#1) \textbf{DO}
  \begin{myenum}\ignorespaces
    \BODY
  \end{myenum}
}

\NewEnviron{algelse}
{\textbf{ELSE}
  \begin{myenum}\ignorespaces
    \BODY
  \end{myenum}
}

\copyrightyear{2018}
\acmYear{2018}
\setcopyright{acmcopyright}
\acmConference{ISSAC '18}{}{July 16-19, 2018, New York, USA}
\acmPrice{15.00}
\usepackage{bm}

\begin{document}

\title{Desingularization of First Order Linear Difference Systems with Rational
  Function Coefficients} \author{Moulay A. Barkatou}
\affiliation{XLIM UMR 7252 , DMI, University of Limoges; CNRS\\
  123, Avenue Albert Thomas, 87060 Limoges, France}
\email{moulay.barkatou@unilim.fr} \author{Maximilian Jaroschek} \affiliation{
  Technische Universität Wien,\\
  Institut for Logic and Computation\\
  Favoritenstraße 9-11, 1040 Vienna, Austria } \affiliation{
  Johannes Kepler Universität Linz, Institute for Algebra\\
  Altenberger Straße 69, 4040 Linz, Austria } \email{maximilian@mjaroschek.com}
\authornote{The author is supported by the ERC Starting Grant 2014 SYMCAR 639270
  and the Austrian research projects FWF Y464-N18 and FWF RiSE
  S11409-N23.}

\fancyhead{}

\renewcommand{\shortauthors}{M.A. Barkatou, M. Jaroschek}

\begin{abstract}
  It is well known that for a first order system of linear difference equations
  with rational function coefficients, a solution that is holomorphic in some
  left half plane can be analytically continued to a meromorphic solution in the
  whole complex plane. The poles stem from the singularities of the rational
  function coefficients of the system. Just as for differential equations, not
  all of these singularities necessarily lead to poles in solutions, as they
  might be what is called removable. In our work, we show how to detect and
  remove these singularities and further study the connection between poles of
  solutions and removable singularities. We describe two algorithms to
  (partially) desingularize a given difference system and present a
  characterization of removable singularities in terms of shifts of the original
  system.
  % and the extension of
  % numerical sequences at these points.
\end{abstract}


\maketitle
\section{Introduction}
\label{sec:intro}
First order linear difference systems are a class of pseudo-linear
systems~\cite{bronstein,Bar06,jacobson} of the form $\T(Y)=AY$, where $\T$ is the
forward- or backward shift operator and $A$ an invertible matrix with, in our case, rational
function coefficients. To study properties of possible solutions~$Y,$ it is not
always necessary to explicitly compute the solution space, but one can rather
obtain the information from the system itself. Properties that can be derived in
this fashion comprise, among others, the asymptotic behavior~\cite{BaCh1, BaCh2,flajolet},
positive/negative (semi-) definiteness~\cite{abh,zhang}, holomorphicity, and
closure properties of (a class of) solutions~\cite{tetrahedron}.

In the center of attention when analyzing difference and differential systems
lie the poles of the rational function coefficients.  %MB %For example, in the differential case, the radius of convergence of the Taylor expansion around the origin of a solution holomorphic at zero is determined by the absolute value of the singularity closest to the origin.
It is well known that, like in the case of differential equations and systems,
not all poles of the coefficients of a difference system lead to %MB %poles of
singularities for solutions. These apparent singularities can therefore distort the properties of
solutions and should be circumvented in the analysis. One technique to do so is
{\em desingularization}---transforming a given system (or operator) in a way that
%MB %preserves the original solutions but 
removes as many poles of the system as
possible to discard apparent singularities. In this paper we describe the first
algorithm to desingularize first order linear difference systems with rational
function coefficients. Our main tool in the treatment of these systems are
polynomial basis transformations. We show how to achieve desingularization by
composing several basic and easy to compute transformations, and our procedure
results in the provably ``smallest possible'' such desingularizing
transformation in the sense that any other desingularizing transformation can be
obtained as a right multiple.

The main contributions of this paper are:
\begin{enumerate}
\item The first algorithm to desingularize---partially, or, if possible,
  completely---first order linear difference systems with rational function
  coefficients.
\item A non-trivial necessary and sufficient condition for a given system to be desingularizable at a given singularity. %MB %the existence of removable singularities.
\item With the help of (2), an analysis of the connection between removable and
  apparent singularities of difference systems and their meromorphic function
  solutions.
\item An algorithm for reducing the rank of the leading matrix at a
  singularity of a linear difference system.
% \item formulation of the system analogue of the extension of numerical sequences
%   beyond removable singularities.
\end{enumerate}

\indent In the context of single linear difference equations~\cite{abh,ah},
linear differential equations~\cite{tsai} and, more general, Ore
operators~\cite{chen:curve,kauers,jaroschek}, desingularization and the effects
of removable singularities have been extensively studied in recent
years. In~\cite{zhang}, the author presents an extension of the idea of
desingularization that also takes into account the leading number coefficients
of Ore operators. For first order differential equations, a first algorithm for
desingularization was given in~\cite{bm}. 

It is possible to convert any first order linear difference system to a
difference operator of higher order and vice versa~\cite{birkhoff, BaCh1, BaCh2,
  barkatou}. Desingularization of systems could therefore be done by computing
for a given system the corresponding operator, use existing techniques to
desingularize the operator and then constructing the desingularized system from
the new operator. While this is possible, the procedure comes with at least two
caveats:
\begin{enumerate}
\item It can be observed that the coefficients grow very large in the conversion
  process which has severe negative impact on the computation time.
\item Desingularization on operator level is done by finding a suitable
  left-multiple of the given operator. In general, this leads to an increase in
  order, and thus to an increase in the dimension of the solution
  space.
  % Desingularization on the system level leaves the solution space unchanged.
\end{enumerate}
Both problems are avoided when dealing directly with systems instead of
operators, making the results presented in this paper an essential tool for
analyzing difference systems.

The paper is organized as follows. In Section~\ref{sec:prelim} we remind the
reader of the formal definition of linear difference systems with rational
function coefficients, well known results about meromorphic function solutions
and the notion of apparent singularities. In Section~\ref{sec:desing}, we
present an algorithm to remove poles of difference systems and give a necessary
and sufficient condition for a singularity to be removable. 
% The removal of poles of inverse systems is studied in Section~\ref{sec:inverse},
% and 
Lastly, the connection between removable poles and apparent singularities is
established in Section~\ref{sec:appsing} before concluding the paper in
Section~\ref{sec:conclusion}.

% we address the problem of how to
% extend numerical sequences that satisfy a linear difference system beyond a
% removable singularity in Section~\ref{sec:extend} 

\section{Difference Systems and Removable Singularities}
\label{sec:prelim}

Let $\F$ be a subfield of the field $\MC$ of complex numbers,
$\F(z)$ the field of rational
functions over $\F$ and $\T$ the $\F$-automorphism of $\F(z)$ defined by
$\T (z)=z+1$.  A homogeneous system of first-order linear difference equations
with rational function coefficients is a system of the form
\begin{equation}
  \label{eq:system}
  \T (Y) = A Y,
\end{equation}
where $Y$ is an unknown $d$-dimensional column vector, $\T(Y)$ is defined
component-wise, and $A$ is an element of $\operatorname{GL}_d(\F(\n))$, the
group of invertible matrices of size $d\times d$ with entries in $\F(\n)$. We
denote the set of matrices of size $d\times d$ with entries in $\F[\n]$ as
$\operatorname{Mat}_d(\F[\n])$.  A (block) diagonal matrix with
entries (respectively blocks) $a_1,\dots,a_d$ is denoted by
$\operatorname{diag}(a_1,\dots,a_d)$. We will refer to system~\eqref{eq:system}
as~$[A]_\T$.

Given a matrix $T\in\operatorname{GL}_d(\F(\n))$, we can apply a basis
transformation
\[Y=TX,\]
and substitute $TX$ into system~\eqref{eq:system} to arrive at an equivalent
system
\[\T (X)=T[A]_\T X,\]
where $T[A]_\T$ is defined as
\[T[A]_\T:=\T (T^{-1})AT.\]
A difference system $[A]_\T$ can be rewritten as
\begin{equation}
  \label{eq:invsystem}
  \IT (Y) =  \IA Y,
\end{equation}
where $\IA:=\IT (A^{-1})$. We will refer to system~\eqref{eq:invsystem} as
$[\IA]_{\IT}$. A transformation $Y=TX$ yields the equivalent
system 
\[\IT(X)=T[A^*]_{\IT}X,\]
with 
\[T[A^*]_{\IT}:=\T^{-1}(T^{-1})\IA
T.\]
The set of meromorphic solutions of $[A]_\T$ form a vector space of dimension
$d$ over the field of $1$-periodic meromorphic functions. It is well
known~\cite{praagman} that any difference system $[A]_\T$ possesses a
fundamental matrix of meromorphic solutions.  If $F$ is a holomorphic solution
of~\eqref{eq:system} in some left half plane ($\operatorname{Re} \n < \lambda$
for some $\lambda\in\set R$),
then it can be analytically continued to a meromorphic solution in the whole
complex plane $\MC$ using the relations:
\begin{alignat*}2
  F(\n) &= \IT(A) \T^{-2}(A)\cdots \T^{-n}(A)\T^{-n}(F)(\n)\\
  ~&=A(\n - 1)A(\n - 2)\cdots A(\n - n)F (\n - n),
\end{alignat*}
which are valid everywhere except at the points of the form $\p + n$ where $\p$
is a pole of $A$ and $n$ is a positive integer ($n = 1, 2, \dots$). If $F$
is a holomorphic solution of~\eqref{eq:system} in some right half plane
($\operatorname{Re} \n > \lambda$), then it can be analytically continued to a
meromorphic solution in the whole complex plane $\MC$ using the relations:
\begin{alignat*}2
  F(\n) &= \T(A^*)\T^2(A^*)\cdots \T^n(A^*) \T^nF(\n)\\
  ~&=A^{*}(\n+1)A^{*}(\n + 2)\cdots A^{*}(\n + n)F(\n + n),
\end{alignat*}
which are valid everywhere except at the points of the form $\p - n$ where $\p$
is a pole of $A^{*}$ and $n$ is a positive integer ($n = 1, 2,\dots$).

We will denote by ${\mathcal P}_r(A)$ (respectively ${\mathcal P}_l(A)$) the set of
poles of $A$ (respectively $A^{*}$).  The elements of ${\mathcal P}_r(A)$
(respectively ${\mathcal P}_l(A)$) will be called the r- (respectively l-)
singularities of the system~\eqref{eq:system}. 
% The set of finite singularities
% of~\eqref{eq:system} is ${\mathcal S}(A):={\mathcal P}_r(A) \cup {\mathcal P}_l(A)$.  
A
point $\p \in \MC$ is said to be {\em congruent} to a given r- (respectively l-)
singularity $\p_0$ of $[A]_\T$ if $\p=\p_0+k$ (respectively $\p=\p_0-k$)
for some positive integer $k$.

The finite singularities of the solutions of $[A]_\T$ are among the points that
are congruent to the singularities of the system.

\begin{definition}\label{def:appsing}
  Let $\p$ be a pole of $A$ (respectively pole of $A^{*}$). It is called
  \begin{enumerate}
  \item a removable r- (respectively l-) singularity if any solution of $[A]_\T$
    which is holomorphic in some left (respectively right) half-plane can be analytically
    continued to a meromorphic solution which is holomorphic at $\p+1$
    (respectively $\p-1$). 
    %M %In this case, we say that $[A]_\T$
    % (respectively~$[A^{*}]_\T$)  is desingularizable at $\p$ and $\p$ is removable  from $[A]_\T$.
    % ($[A^{*}]_\T$).
  \item an apparent r- (respectively l-) singularity if any solution of $[A]_\T$ which is
    holomorphic in some left (respectively right) half-plane can be analytically
    continued to a meromorphic solution which is holomorphic at each point of
    $\p +\N^*$ (respectively~$\p-\N^*$). 
  \end{enumerate}
\end{definition}

\begin{example}
  \label{ex:mainex}
  A $2\times 2$ system of linear difference equations is given by
  \[Y(\n+1)=AY=\begin{pmatrix}0 & 1 \\
        \frac{-2(\n +1)}{\n-2} & \frac{3(\n -
          1)}{\n-2}\end{pmatrix}Y(\n),\quad A^*=\begin{pmatrix}
        \frac{3(z-2)}{2z} & \frac{3-z}{2z}\\ 1 & 0
      \end{pmatrix}.\]
  Here ${\mathcal P}_r(A)=\{2\}$ and the points that are congruent to $\p=2$ are
  $3, 4, 5, \dots$. We have ${\mathcal P}_l(A)= \{0\}$ and the corresponding
  congruent points are $-1,-2,-3,\dots$. It can be easily verified that a
  fundamental matrix of solutions of this system is given by
  \[F(\n)=\begin{pmatrix}2^\n & \n^3 + 5\n + 6\\
      2^{\n+1}\quad & \n^3 + 3\n^2 + 8\n + 12\end{pmatrix}.\]
  % We remark that all the entries of these fundamental matrix are holomorphic in
  % the whole complex plane.
\end{example}

We focus on studying r-singularities. L-singularities can be
removed in the same way by considering $A^*$ and $\phi^{-1}$ instead of $A$
and~$\phi$.

We give an algebraic characterization of removable singularities. 
Let $q \in \F[z]$ be an irreducible polynomial. For $f \in \F(z) \setminus \{0\}$,
we define $\ord_q(f)$ to be the integer $n$ such that $f=q^n\frac{a}{b}$, with
$a, b \in \F[z] \setminus \{0\}$, $q \nmid a$ and $q \nmid b$. We put
$\ord_q(0) = +\infty$.  Let $ \SF_q = \{ f \in \F(z) : \ord_q(f) \geq 0 \}$ be
the {\it local ring} at $q$ and $\SF_q/q\SF_q$ the residue field of
$\F(z)$ at $q$. Let $\pi_q$ denote the canonical homomorphism from $\F[z]$ onto
$\F[z]/\langle q\rangle$.  It can be extended to a ring-homomorphism from
$\SF_q$ onto $\F[z]/\langle q\rangle$ as follows: let $f \in \SF_q$; by
definition of $\SF_q$, $f$ can be written $f=a/b$ where $a,b\in \F[z]$ and
$q \nmid b$. We can find $u, v \in \F[z]$ such that $ub+vq = 1$, {\it the value}
of $f$ at $q$, denoted by $\pi_q(f)$, is then defined as $\pi_q(ua)$.  Sometimes
we write $ f \m q$ for $\pi_q(f)$.  It is clear that $\pi_q$ is well-defined on
$\SF_q$ and is a surjective ring-homomorphism. The kernel of $\pi_q$ is
$q\SF_q$, so $\SF_q/q\SF_q$ and $\F[z]/\langle q\rangle$ are isomorphic.
  
If $A=(a_{i,j})$ is a finite-dimensional matrix with entries in $\F(z)$, we
define the {\em order} at $q$ of $A$ by
$ \ord_q (A) := \min_{i,j}{(\ord_q (a_{i,j}))}.$ We say that $A$ has a pole at
$q$ if $\ord_q (A)<~0$.  We define the {\it leading matrix} of $A$ at $q$
(notation $\lc_q(A)$) as the leading coefficient $A_{0,q}$ in the $q$-adic
expansion of $A$:
\[A = q^{\ord_q(A)}( A_{0,q} + qA_{1,q} +  q^2A_{2,q} + \dots ).\]
Here the coefficients $A_{i,q}$ are matrices with entries in the field
$\F[z]/\langle q\rangle$.  Note that the matrix $A_{0,q}$ is the value of the
matrix $q^{-\ord_q(A)}A$ at $q$.

For a rational function $r=p/q$ with $p$ monic and $\gcd(p,q)=1$, we write
$\num(r):=p$ and $\den(r):=q$. Similarly, for a matrix
$M \in \operatorname{Mat}_d(\F(\n))$ we denote by $\den(M)$ the common
denominator of all the entries of $M$ and denote by $\num(M)$ the polynomial
matrix $ \num(M) := \den(M) M$.

\begin{definition}
  Let $A\in\operatorname{GL}_d(\F(z))$ and $q\in \F[z]$ be an irreducible
  polynomial. We say that $q$ is a $\T$-minimal pole of $A$ if $q \mid \den(A)$
  and for all $j\in \N^*$, $\T^j(q) \nmid \,\den(A)$.
\end{definition}

We can now give an algebraic definition of desingularizability of difference
systems in an inductive fashion.
%MB
\begin{definition}
  \label{def:desing}
  Let $A\in\operatorname{GL}_d(\F(z))$ and let $q\in \F[z]$ be an irreducible  %$\T$-minimal
  pole of~$A$. 
  \begin{enumerate}
  \item If $q$ is $\T$-minimal, we say that the system $[A]_\T$ is partially desingularizable at $q$ if there exists a polynomial transformation
  $T\in\operatorname{GL}_d(\F(z))\cap\operatorname{Mat}_d(\F[z])$ such that
  $\ord_q(T[A]_\T) > \ord_q(A)$ and $\ord_p(T[A]_\T) \geq \ord_p(A)$ for any
  other irreducible polynomial $p \in \F[z]$.  If moreover,  $\ord_q(T[A]_\T) \geq 0$ then we say that $[A]_\T$ is
  desingularizable at $q$ and we call~$T$~a
  desingularizing transformation for $[A]_\T$ at~$q$.

\item If $q$ is not $\T$-minimal, then we call $[A]_\T$ (partially)
  desingularizable at $q$ if there exists a desingularizing transformation $T$
  for all poles of $A$ of the form $\T^k(q)$, $k\geq 1$, and $T[A]_\T$ is either (partially) desingularized at $q$ or
  (partially) desingularizable at
  $q$. % at $q$ by some $\tilde{T}$, where $T$ is a
% desingularizing transformation for $\T^{-k}(q)$,
% $k=\min\{k'\mid k\in\set N^+,\ord_{\T^{-k'}(q)}(A)<0\}$, such that
% $\ord_{\T^{-k}(q)}(T[A])$ is minimal among all such transformations. In
% particular, $[A]_\T$ is called desingularizable at $q$ iff $B$ is
% desingularizable at $q$ by $\tilde{T}$ and $\ord_{\T^{-k'}(q)}(A)>0$.  A
% desingularizing transformation for $[A]_\T$ at $q$ then is~$\tilde{T}\cdot T$.
\end{enumerate}
\end{definition}

While it is immediate that, for a $\T$-minimal pole $q$, the algebraic notion of
desingularization implies that the roots of $q$ are removable in the sense of
Definition~\ref{def:appsing}, the converse is not obvious and is proven later in
Section~\ref{sec:appsing}. Consequently, the roots of $q$ are apparent
singularities if (and only if) the system $A$ is desingularizable at all poles
of $A$ of the form $\T^{-k}(q)$, $k\geq 0$. In practice, in order to desingularize
a system at a non-$\T$-minimal pole $q$, one first removes the $\phi$-minimal
pole congruent to $q$. The resulting system then has a new $\phi$-minimal pole
'closer' to $q$. One can repeat this process until $q$ itself is $\phi$-minimal
and eventually removed. A desingularizing transformation for $q$ is then given by the product of all
the transformations obtained during this process.

Let us illustrate in the next example
why we require removing all singularities left of a given pole, thus making it
$\phi$-minimal, before considering it eligible for desingularization.

\begin{example}
  The system $[A]_\T$ given by \[A=\operatorname{diag}\biggl(\frac{(z+1)^2}{z},\frac{1}{z+1}\biggr),\]
can be transformed via $T=\operatorname{diag}(z,1)$ to
$T[A]_\T=\operatorname{diag}(z+1,\frac{1}{z+1})$. The transformed systems still does
not enable analytic continuation at  %MB % $-1$
$0$ of solutions that are holomorphic in the
left half-plane with $\operatorname{Re}(z)<0$.
\end{example}

% We now first address the problem of constructing desingularizing transformations
% for a given system $[A]_\T$.

\section{Removing r-Singularities}
\label{sec:desing}
\pdfstringdefDisableCommands{%
    \renewcommand*{\bm}[1]{#1}%
    % any other necessary redefinitions 
}
\subsection{$\bm{\T}$-Minimal Desingularization}

We begin our discussion of removing r-singularities by deriving a method for
shifting a factor in the denominator of a given system in a way that
allows, if possible, cancellation with zeroes of the system. For this we bring the leading
matrix of $A$ into a specific form.

\begin{lemma}
  \label{lem:smithnf}
  Let $A$ be a $d\times d$ matrix with entries in $\F(\n)$ and let
  $q\in \F[z]$ be an irreducible pole of $A$. Set $n:= -\ord_q(A)$ and $r:=\ra(\lc_q(A))$, the
  rank of the leading matrix of $A$ at $q$. There exists a unimodular polynomial
  transformation $S$ such that $S[A]_\T$ is of the form
  \begin{equation}
    \label{eq:smith}
    \begin{pmatrix}
      \frac{1}{q^n}A_1 & \frac{1}{q^{n-1}}A_2,
    \end{pmatrix},
  \end{equation}
  where $A_1, A_2$ are matrices with entries in $\SF_q$ of size $d\times r$ and
  $d\times d-r$ respectively with $\ra(A_1) =r$.
\end{lemma}

\begin{proof}
  The leading matrix $\lc_q(A)$ of $A$ at $q$ is a matrix with entries in the
  residue field $\F[\n]/\langle q\rangle$.  There exists a non-singular matrix
  $Q$ with entries in $\F[\n]/\langle q\rangle$ such that $\lc_q(A)\cdot Q$ is
  in a column-reduced form, i.e.\ the last $d-r$ columns of $\lc_q(A)\cdot Q$
  are zero, and $Q$ is of the form $Q=C\cdot(\id_d+U)$, where $C$ is a
  non-singular constant matrix, $I_d$ the identity matrix of dimension
  $d\times d$, and $U$ is a strictly upper triangular matrix. Taking $S=Q$ as a
  matrix in $\operatorname{Mat}_d(\F[\n])$ will result in $S[A]_\T$ as desired.
\end{proof}

\begin{example}[Example~\ref{ex:mainex} continued]
  \label{ex:mainex2}
  If we set $q=\n-2$, then the leading matrix of the system in
  Example~\ref{ex:mainex} at $q$ is
  \[ \lc_q(A)=(qA) \m q = \begin{pmatrix}0 & 0\\
      -6 & 3\end{pmatrix}.\]
  A suitable transformation to bring this matrix into a column-reduced form is
  \[S=\begin{pmatrix}1 & \frac{1}{2} \\ 0 & 1\end{pmatrix}.\]
  Applying $S$ to $A$ gives
  \[S[A]_\T = \begin{pmatrix}\frac{\n+1}{\n-2} & 0 \\ \frac{-2\n-2}{\n-2} &
      2\end{pmatrix}.\]
\end{example}

\begin{lemma}
  \label{lem:shearing}
  Let $A\in\operatorname{GL}_d(\F(z))$ and $q\in \F[z]$ be a $\T$-minimal pole
  of $A$.  Suppose that $A$ is of the form~\eqref{eq:smith} and let
  $r=\ra(\lc_q(A))$.  If $[A]_\T$ is partially desingularizable at $q$ then any
  desingularizing transformation $T$ for $[A]_\T$ can be written as
  \[T = D\cdot \tilde{T},\text{ where}\]
    \[D=\operatorname{diag}(\underbrace{q,\dots,q}_{r \text{
        times}},\underbrace{1,\dots,1}_{d-r \text{ times}})\text{ and
    }\tilde{T}\in\operatorname{GL}_d(\F(\n))\cap\operatorname{Mat}_d(\F[z]).\]
\end{lemma}

\begin{proof}
  Let $n=-\ord_q(A)$. Suppose we are given a desingularizing transformation
  $T \in\operatorname{GL}_d(\F(\n))$ and let $B=T[A]_\T$. Then we have that
  $\T(T)B=AT$ and hence \[\T(T) (q^{n}B )= (q^{n}A )T.\]
  Since $\ord_q(q^{n}B)>0$ and the orders at $q$ of the other matrices involved
  in the equality are non-negative, we get
  \[\pi_q(q^{n}A )\pi_q(T) = 0.\]
  By assumption, the matrix $\pi_q(q^{n}A) $ is of the form
  \[\pi_q(q^{n}A )=\begin{pmatrix} A_1 & 0\end{pmatrix},\]
  where $A_1$ is a $d\times r$ matrix with linearly independent columns. Thus,
  the first $r$ rows of $\pi_q(T)$ must be zero, i.e \ the first $r$ rows of $T$
  have to be divisible by $q$. This yields the claim.
\end{proof}

\begin{remark}
  \label{rem:shearing}
  Note that the determinant of any desingularizing transformation $T$ of $A$ at
  $q$, not necessarily $\T$-minimal, is divisible by $q$; in fact
  $q^r \mid \det(T)$. It then follows that $\T(q)$ divides $\det(\T(T))$; in
  fact $\T(q)^r \mid \det(\T(T))$.
\end{remark}

\begin{lemma} 
  \label{lem:dispersion}
  Let $A\in\operatorname{GL}_d(\F(z))$ and let $q\in \F[z]$ with
  $q\mid\operatorname{den}(A)$ be an irreducible pole.  If $[A]_\T$ is
  (partially) desingularizable at $q$ then there exists a maximal positive
  integer~$\ell$ such that $\T^\ell(q)\mid \operatorname{num}(\det(A))$.
\end{lemma}

\begin{proof} First, suppose $q$ is $\T$-minimal. There are only finitely many
  factors of $\num(\det(A))$ of positive degree because $A$ is
  non-singular. Thus it suffices to show that there exists a positive integer
  $\ell_0$ such that $\T^{\ell_0}(q)\mid \num(\det(A))$. Let $T$ be a
  desingularizing transformation of $A$ at $q$. Put $B:=T[A]_\T$ and denote
  $\det(T)$ by $t$.  Then, due to the desingularization property, we have that
  \[ \T (T) ^{-1} \num(A) T = \frac{\den(A)}{\den(B)} \num(B) \in
    \operatorname{Mat}_d(\F[\n]).\] Hence
  \[\frac{\det(\num(A))t}{\T(t)} \in \F[\n].\]
  Let $\ell_0$ be the largest integer such that $\T^{\ell_0}(q)\mid \T(t)$. By
  Remark~\ref{rem:shearing}, $\ell_0$ is strictly positive. Since
  $\T^{\ell_0}(q)\nmid t$, it follows that $\T^{\ell_0}(q)\mid
  \det(\num(A))$.
  Now from the relation $\det(\num(A)) = \den(A)^d\det(A)$ and since we assumed
  that $\den(A)$ has no factor of the form $\T^{j}(q)$ with $j \in \N^*$ we can
  conclude that $\T^{\ell_0}(q)\mid \operatorname{num}(\det(A))$.
  \noindent To see that the theorem holds for non-$\T$-minimal poles, let
  $\tilde q$ be a non-$\T$-minimal pole congruent to $q$, i.e.\ there exists a positive integer
  $k$ such that $\T^k(\tilde q)=q$. Then
  $\T^{k+\ell_0}(\tilde q)=\T^{\ell_0}(q) \mid \operatorname{num}(\det(A))$.
\end{proof}

\begin{definition}
  \label{def:dispersion}
  Let $A\in\operatorname{GL}_d(\F(z))$ and $q\in \F[z]$ be an irreducible pole
  of $A$. We define the $\T$-dispersion of $A$ at $q$ as :
  \[\T\textit{-dispersion}(A,q) = \max{ \{\ell \in \N^* \,\text{ s.t. }\T^\ell(q) \mid
      \num(\det(A)) \}}.\]
  When the latter set is empty we put $\T$-dispersion$(A,q) =0$.
\end{definition}

Note that a necessary condition that $[A]_\T$ can be (partially) desingularized
at $q$ is that $\T$-dispersion$(A,q) >0$.

\begin{example}[Example~\ref{ex:mainex2} continued]
  \label{ex:mainex3}
  The determinant of $S[A]_\T$ in Example~\ref{ex:mainex2} is
  $\frac{2(\n+1)}{\n-2}$. Therefore the $\T$-dispersion of $S[A]_\T$ at $q=z-2$
  is equal to 3.
\end{example}

% Before compiling the transformations given in Lemma~\ref{lem:smithnf} and
% Lemma~\ref{lem:shearing} into a desingularization algorithm, we first have to
% make a decision about the order in which congruent poles should be removed, as
% the removal of one pole can influence the removability of another

% \begin{example}
% \label{ex:order}
% The scalar system $\T(Y(\n))=\frac{\n+2}{\n(\n+1)}Y(\n)$ can either be desingularized
% at $q=\n+1$ or at $\tilde{q}=\n$:
% \begin{alignat*}2
% &T[A] = \frac{1}{\n},& &\text{ for }T=\n+1,\\
% &\tilde{T}[A] = \frac{1}{\n+1},& &\text{ for }\tilde{T}=\n\T(\n).
% \end{alignat*}
% Desingularization at both poles simultaneously is
% not possible. Note, however, that 
% \[S[T[A]] = \frac{1}{\n+1}=\tilde{T}[A],\text{ for }S=\n,\]
% so when first removing the $\T$-minimal pole, we can still arrive at the
% system given by $\tilde{T}[A]$ by shifting the pole in $T[A]$ to the left, while shifting
% the pole in $\tilde{T}[A]$ to the right to arrive at $T[A]$ is not possible with
% a polynomial transformation.
% \end{example} 

% The Definition~\ref{def:appsing} of apparent singularities and the observation
% in Example~\ref{ex:order} suggest that it is beneficial to prioritize
% desingularization with respect to $\T$-minimal poles.

We will now describe an algorithm for desingularizing a given system $[A]_\T$ at
a $\T$-minimal pole $q$. By repeatedly applying the algorithm to $[A]_\T$, it is
then possible to desingularize the system at all removable singularities.
% More precisely, we address the following problem:
% \noindent given a matrix $A\in\operatorname{GL}_d(\F(z))$ and an irreducible
% polynomial $q\in \F[z]$ such that $\displaystyle{q \mid \den(A)}$, construct, if
% possible, a polynomial transformation $T \in \operatorname{GL}_d(\F(z))$ such
% that
% \[\den(T[A]_\T) \mid \den(A)\text{ and }\ord_q(T[A]_\T) > \ord_q(A),\]
% and $\ord_p(T[A]_\T) \geq \ord_p(A)$ for any other irreducible polynomial
% $p \in \F[z]$. 
It is sufficient to treat the case where $q$ is a single
and simple pole of~$A$ (i.e.\ $qA$ has polynomial entries). This is stated in
the following lemma.

\begin{lemma}\label{lem:simpole}
  Let $A\in\operatorname{GL}_d(\F(z))$ and let $q\in \F[z]$ be a 
  $\T$-minimal pole of $A$. Set $h=\frac{\den{(A)}}{q}$ so that the
  matrix $hA = q^{-1}\num(A)$ has a single and simple pole at $q$. Then the
  system $[A]_\T$ is (partially) desingularizable at~$q$ if and only if the system
  $[hA]_\T$ is desingularizable at~$q$. More precisely, a polynomial matrix
  $T \in \operatorname{GL}_d(\F(z))$ is a desingularizing transformation for
  $[hA]_\T$ at $q$ if and only if $T$ (partially) desingularizes $[A]_\T$ at~$q$.
\end{lemma}

\begin{proof} 
  It is a direct consequence of the following (trivial but interesting)
  property: for all $T \in \operatorname{GL}_d(\F(z))$ and
  $h \in \F[z]\setminus \{0\}$, one has $T[(hA)]_\T = h\cdot(T[A]_\T)$.
\end{proof}

% In the case of not-$\T$-minimal poles, only one implication holds. Taking $A$ as
% \[A=\begin{pmatrix} \frac{1}{\n} & 0 \\ 0 & \frac{\n+2}{\n+1}\end{pmatrix},\quad
%   q=\n\text{ and }h=\n+1,
% \]
% shows that the
% desingularizability of $[hA]_\T$ at $q$ does not imply that $[A]_\T$ is
% desingularizable at $q$, as the multiplication by $h$ adds zeroes to the system
% that allow cancellation with $q$ that are not possible in $A$.

\begin{remark} 
  With the notation of the above lemma, the $\T$-dispersion of $[hA]_\T$ at $q$
  is greater than or equal to the $\T$-dispersion of $[A]_\T$ at~$q$, and equality
  holds if $q$ is $\T$-minimal. It follows from the fact that
  $\det (hA) =h^d\cdot\det(A)$.
\end{remark}

\begin{lemma}
  \label{lem:reduce}
  Let $A\in\operatorname{GL}_d(\F(z))$. Suppose that $A$ has a single, simple,
  irreducible pole at
  $q$. If $[A]_\T$ is desingularizable at $q$ with $\T$-dispersion $\ell$, then
  there exist a unimodular polynomial matrix~$S$ and a diagonal polynomial
  matrix $D$ such that $(S\cdot D)[A]_\T$ is either desingularized (with respect
  to $q$) or desingularizable at $\T(q)$ with $\T$-dispersion $\ell-1$.
\end{lemma}

\begin{proof}
  We first take $S$ as in Lemma~\ref{lem:smithnf} so that $S[A]_\T$ has the form
  \[S[A]_\T= 
      \begin{pmatrix} {\frac{\tilde{A}_{1,1}}{q}}&
        {\tilde{A}_{1,2}}\\{\frac{\tilde{A}_{2,1}}{q}} &{\tilde{A}_{2,2}}
      \end{pmatrix} ,\]
  where the $\tilde{A}_{i,j}$ are blocks with polynomial entries, the diagonal
  blocks are of size $r=\ra{(\lc_q(A))}$ and $d-r$ respectively.  Take $D
  =\operatorname{diag}(q\id_r,\id_{d-r})$ as in Lemma~\ref{lem:shearing}.  Then
  the matrix $B := (S\cdot D)[A]_\T$ has the form
  \[B= 
      \begin{pmatrix} {\frac{\tilde{A}_{1,1}}{\T(q)}}&
        {\frac{\tilde{A}_{1,2}}{\T(q)}}\\{\tilde{A}_{2,1}} &{\tilde{A}_{2,2}}
      \end{pmatrix}.\]
  The resulting system $[B]_\T$ has at worst a simple and
  single pole at $\T(q)$ with $\T$-dispersion $\ell-1$.
\end{proof}

\begin{example}[Example~\ref{ex:mainex3} continued]
  \label{ex:mainex4} 
  The rank of the leading matrix in Example~\ref{ex:mainex2} is~1. We apply the
  transformation \[D_1=\begin{pmatrix}\n-2\;\;\; & 0\\ 0 & 1\end{pmatrix},\]
  to $S[A]_\T$ of Example~\ref{ex:mainex3} and arrive at the system
  \[(S\cdot D_1)[A]_\T =\begin{pmatrix}\frac{\n+1}{\n-1} & 0\\ -2\n-2\;\;\; &
      2\end{pmatrix}.\]
  The determinant of $(S\cdot D_1)[A]_\T$ is $\frac{2(\n+1)}{\n-1}$. The new
  $\T$-dispersion is~2.
\end{example}

\begin{theorem}
  \label{thm:desingularization}
  Let $A$ be desingularizable at a single, simple, irreducible pole~$q$. Then
  there exists an integer $n$, unimodular polynomial matrices $S_1,\dots,S_n$
  and diagonal polynomial matrices $D_1,\dots,D_n$ such that
  \[T = S_1\cdot D_1\cdots S_n\cdot D_n,\]
  is a desingularizing transformation for $A$ at $q$. Furthermore, any other
  desingularizing transformation $T'$ for $A$ at $q$ can be written as
  \begin{equation}
    \label{eq:leftfactor}
    T' = T\cdot\tilde{T}\text{ with }\tilde{T}\in\operatorname{GL}_d(\F(\n))\cap\operatorname{Mat}_d(\F[z]).
  \end{equation}
\end{theorem}

\begin{proof}
  By Lemma~\ref{lem:dispersion}, a desingularizable system $[A]_\T$ has strictly
  positive $\T$-dispersion~$\ell$. Applying the transformation $S\cdot D$ as in
  Lemma~\ref{lem:reduce} gives a system equivalent to~$[A]_\T$ having at
  worst a pole at $\T(q)$ (instead of $q$) but with reduced $\T$-dispersion. After at
  most~$\ell$ such transformations, the resulting matrix $T[A]_\T$ has to be
  desingularized at~$q$. This shows that $T$ can be chosen as in the statement
  of the theorem. To see that any other desingularizing transformation $T'$ of
  $[A]_\T$ at $q$ can be written as in~\eqref{eq:leftfactor}, we first note that
  since $S_1$ is unimodular, for any such $T'$ we have
  \[T'=S_1\cdot\underbrace{(S_1^{-1}\cdot
      T')}_{\mathclap{=:T''\in\operatorname{GL}_d(\F(\n))\cap\operatorname{Mat}_d(\F[z])}},\]
  and therefore we can assume that $A$ is of the form~\eqref{eq:smith}. Then, as
  was shown in Lemma~\ref{lem:shearing}, we can write
  \[T''= D_1\cdot\tilde{T},\]
  with $\tilde{T}\in\operatorname{GL}_d(\F(\n))\cap\operatorname{Mat}_d(\F[z]))$. Again, we can repeat this
  reasoning $n$ times until we arrive at the desired form.
\end{proof}

\begin{example}[Example~\ref{ex:mainex4} continued]
  \label{ex:mainex5}
  The leading matrix of $(S\cdot D_1)[A]_\T$ as in Example~\ref{ex:mainex4} at
  $\T(q)=\n-1$ is already in column-reduced form and of rank 1. We apply the
  transformation $D_2=\operatorname{diag}(\n-1,1)$, and get
  \[(S\cdot D_1\cdot D_2)[A]_\T =\begin{pmatrix}\frac{\n+1}{\n} & 0\\ -2\n^2 +
      2\;\;\; & 2\end{pmatrix}.\]
  Again, the leading matrix of this system at $\T^2(q)=\n$ is column-reduced and
  of rank 1. Finally, after applying the transformation $D_3=\operatorname{diag}(\n,1),$
  we get the desingularized system
  \[(S\cdot D_1\cdot D_2\cdot D_3)[A]_\T =\begin{pmatrix}1 & 0\\ -2\n^3 +
      2\n\;\;\; & 2\end{pmatrix}.\]
  Collecting all the transformations, we see that a desingularizing
  transformation for $A$ at $q=\n-2$ is given by
  \[T = S\cdot D_1\cdot D_2\cdot D_3 = \begin{pmatrix}\n^3 - 3\n^2 + 2\n\;\;\; &
      \frac{1}{2}\\ 0& 1\end{pmatrix}.\]
\end{example}

As was already shown in Lemma~\ref{lem:dispersion}, a positive $\T$-dispersion
is a necessary condition for a removable singularity. For a given system $[A]_\T$
and an irreducible polynomial $q$, the $\T$-dispersion can be obtained by
computing the largest integer root of the resultant
$\operatorname{res}_{\n}(q(\n+k),\operatorname{num}(\det(A)))$. This, together
with Theorem~\ref{thm:desingularization} and its proof gives rise to
Algorithm~1.

% \begin{remark}
% The algorithm can be modified to detect whether a system $[A]_\T$ is
%   desingularizable at an irreducible polynomial $q$ by determining the
%   dispersion $\ell$ of $A$ via computing the largest integer root of the
%   resultant
%   $\operatorname{res}_{\n}(q(\n+k),\operatorname{num}(\det(A)))$. Then, if it is
%   not possible to remove $q$ after $\ell$ loop iterations, $A$ is not
%   desingularizable at $q$.\color{red} Maybe add this to the algorithm\color{black}
% \end{remark}
\begin{figure}[h]
  \begin{center}
\begin{algo}
  {desingularize\_A($A,q$)}
  {$A$ with entries in $\F(\n)$ and a single, simple,
    irreducible pole $q\in\F[\n]$.}
  {$(T,T[A]_\T)$ s.t.\ $T[A]_\T$ is desingularized at
    $q$, or $(I_{d},A)$ if desingularization is not possible.}
  \item $T \leftarrow I_{d}$
  \item \begin{while}{$\T$-$\operatorname{dispersion}(A,q)>0$ \textbf{AND} $\operatorname{den}(A) = 0\ \m\ q$}
    \item $A_0\leftarrow   \lc_q(A)$
    \item $S \leftarrow$ as in the proof of Lemma~\ref{lem:smithnf}.
    \item $D\;{\leftarrow} \operatorname{diag}(q,\dots,q,1,\dots,1)$ with
      $\operatorname{rank}(A_0)$ many ${}\qquad$ elements equal to $q$.
    \item $A\leftarrow \T(S\cdot D)^{-1}\cdot A\cdot (S\cdot D)$
    \item $T\leftarrow T\cdot S\cdot D$
    \item  $q\leftarrow \T(q)$
    \end{while}
    \item \textbf{IF} ($\operatorname{den}(A) = 0\ \m\ q$) \textbf{RETURN} $(I_{d},A)$
    \item \textbf{ELSE} \textbf{RETURN} ($T, A$)
    \end{algo}
  \end{center}
  \end{figure}

%  \begin{remark}
% It is easy to see that the determinant of $T$ as computed by Algorithm~1 is
% given by
% \[c\cdot\prod_{i=0}^k\T^{i}q^{m_i}(\n),\quad c\neq 0\in\C, m_i\in\set N^*, k\in\set N^*.\]
% We will make use of this fact in the later discussion.
% \end{remark}

\subsection{Characterization of Desingularizable Poles}

We can give a necessary and sufficient condition for a pole to be
desingularizable. It can be seen as the shift analogue of the nilpotency of the
leading matrix at the considered pole of the system, which is a necessary
condition for an apparent singularity in the differential setting~\cite{bm}.

\begin{proposition} 
  \label{prop:factorial0} 
  Let $q\in \F[z]$ be a $\T$-minimal pole of the system $[A]_\T$.
  Let $\tilde{A} = q^n A$, so that $\ord_q(\tilde{A})=0$ and
  $\pi_q(\tilde{A})=\lc_q(A)$. If $A$ is (partially) desingularizable at $q$
  then there exists a positive integer $k$ such that
  \begin{equation}
    \label{eq:factorial0}
    \pi_q(\tilde{A}\T^{-1}(\tilde{A})\dots \T^{-k}(\tilde{A})) =0.
  \end{equation}
\end{proposition}

\begin{proof}
  Let $T$ be a desingularizing transformation for $[A]_\T$ at $q$ and
  $B=T[A]_\T$.  Then for all non-negative integers $k$ one has
  \[\T(T) B\T^{-1}(B) \dots \T^{-k}(B) = A\T^{-1}(A) \dots \T^{-k}(A)
    \T^{-k}(T),\]
  and hence
  \[\T(T) (q^nB)\T^{-1}(q^nB) \dots \T^{-k}(q^nB) = \tilde{A}\T^{-1}(\tilde{A})
    \dots \T^{-k}(\tilde{A}) \T^{-k}(T).\]
  As $\ord_q(q^nB) >0$ and
  \[ \ord_q(\T^{-j}(q^nB))= \ord_{\T^{j}(q)}(B) \geq \ord_{\T^{j}(q)}(A) \geq 0,
    \; \hbox{for all} \, j \in \N^*,\] we get that
  \[\pi_q(\tilde{A}\T^{-1}(\tilde{A}) \dots \T^{-k}(\tilde{A}) \T^{-k}(T)) =0.\]
  Now we conclude by remarking that for $k$ large enough $\pi_q(T(\n-k))$ is
  invertible.
\end{proof}

We will now show that the factorial relation~\eqref{eq:factorial0} is a sufficient condition for a matrix $A$ to
be partially desingularizable at $q$.

\begin{proposition} 
  \label{prop:factorial1} 
  Let $q\in \F[z]$ be a $\T$-minimal pole of $[A]_\T$. Let $\tilde{A} = q^n A$, so that $\ord_q(\tilde{A})=0$ and
  $\pi_q(\tilde{A})=\lc_q(A)$.
  % Put
  % $n=-\ord_q(A) > 0$ and $\tilde{A}(z) = q^n A(z)$ so
  % that $\pi_q(\tilde{A})=\lc_q(A) \neq 0$.
  If $[A]_\T$ is such that
  the factorial relation~\eqref{eq:factorial0} holds for some integer $k \geq 1$
  then $[A]_\T$ is (partially) desingularizable at $q$.
\end{proposition}

\begin{proof}
  Let $k$ be minimal so that~\eqref{eq:factorial0} holds. Put
  \[M:=\pi_q(\tilde{A}\T^{-1}(\tilde{A}) \cdots \T^{-k+1}(\tilde{A})) \quad
    \text{and} \quad N:= \pi_q(\T^{-k}(\tilde{A})).\]
  By definition of $k$, the matrix $M$ is nonzero (but singular) and we have
  $M \cdot N =0.$ With $d:=\operatorname{dim}(A)$ it follows that \[0 < \ra(M) \leq s:=d-\ra{(N)} < d.\]
  Let $P \in \operatorname{GL}_d(\F[z]/\langle q\rangle)$ such that $P \cdot N$
  has its last $(d-s)$ rows linearly independent over $\F[z]/\langle q\rangle$
  while its $s$ first rows are zero.  Consider the matrix $:U= \T^{k-1}(P^{-1})$
  as an element of $\operatorname{Mat}_d(\F[\n])$ then by applying the
  unimodular transformation $Y=UX$, we can assume that the
  matrix $N$ has the following form:
  \[N=\begin{pmatrix} \zm_s&\zm_{s,d-s} \\{N_{2,1}}&N_{2,2}
             \end{pmatrix},\]
  where $N_{2,1}$ and $N_{2,2}$ are matrices with entries in
  $\F[z]/\langle q\rangle$ of size $(d-s) \times s$ and ${(d-s)}\times (d-s)$
  respectively, so that the last $d-s$ rows of $N$ are linearly independent over
  $\F[z]/\langle q\rangle$.  As $M\cdot N =0$ we have that the $d-s$
  last columns of $M$ are zero.  Let
  $\tilde{A} = (\tilde{A}_{i,j})_{1\leq i,j \leq 2}$ be partitioned in
  four blocks as $N$.  Then $\pi_q(\T^{-k}(\tilde{A}_{1,j})) = 0$ for
  $j = 1, 2$.  In other words, the $s$ first rows of $\tilde{A}$ are
  divisible by $\T^{k}(q)$.  Using the substitution $Y=DX$ where
  $D = \operatorname{diag}(\T^{k-1}(q)\id_s, \id_{d-s})$, we get a new
  system which still has a pole at $q$ of multiplicity at most $n$.  Indeed, we
  have
  \begin{align}
    \label{eq:systemb}
    \begin{split}
  B := \T(D)^{-1}AD = {}& q^{-n}
      \begin{pmatrix}
        {\frac{\T^{k-1}(q)\tilde{A}_{1,1}}{\T^{k}(q)}}&
        {\frac{\tilde{A}_{1,2}}{\T^{k}(q)}}\\{\T^{k-1}(q)\tilde{A}_{2,1}}
        &{\tilde{A}_{2,2}}
      \end{pmatrix}  =  \\
    & q^{-n}
      \begin{pmatrix}
        \T^{k-1}(q)\tilde{A'}_{1,1} & \tilde{A'}_{1,2}\\
        \T^{k-1}(q)\tilde{A}_{2,1} & \tilde{A}_{2,2}
      \end{pmatrix},
    \end{split}
  \end{align}
  for some matrices $\tilde{A'}_{1,1}$, $\tilde{A'}_{1,2}$ with entries in
  $\SF_q$. It is clear that $\den(B)\mid \den(A)$ and that $\ord_q(B) \geq
  \ord_q(A)$. Now we will prove that the factorial relation~\eqref{eq:factorial0}
  holds for $\tilde{B}:=q^nB$ with $k-1$ instead of $k$. For this we remark
  first that
  \[\T(D) \tilde{B}\T^{-1}(\tilde{B}) \cdots \T^{-k+1}(\tilde{B}) =
    \tilde{A}\T^{-1}(\tilde{A}) \cdots \T^{-k+1}(\tilde{A}) \T^{-k+1}(D).\]
  It then follows that
  \[\pi_q(\T(D)) \pi_q(\tilde{B}\T^{-1}(\tilde{B}) \cdots \T^{-k+1}(\tilde{B}))=
    M\cdot \pi_q(\T^{-k+1}(D)).\] We have that
  \[ \pi_q(\T^{-k+1}(D)) = \pi_q(\operatorname{diag}(q\id_s, \id_{d-s})) =
    \operatorname{diag}(\zm_s, \id_{d-s}),\]
  hence $ M\cdot \pi_q(\T^{-k+1}(D)) =0$ (since the $d-s$ last columns of $M$
  are zero).  Now
  $\pi_q(\T(D)) = \pi_q(\operatorname{diag}(\T^k(q)\id_s, \id_{d-s}))$ is
  invertible (since $q$ and $\T^k(q)$ are co-prime), it then follows that
  \[\pi_q(\tilde{B}\T^{-1}(\tilde{B}) \cdots \T^{-k+1}(\tilde{B}))=0.\]
  If $k-1$ is still positive then we can repeat this process for the matrix $B$
  and the polynomial~$q$ until we arrive at $k=1$.  When $k=1$ the above
  factorial relation reduces to $\pi_q(\tilde{B}) =0$ which means that
  $\ord_q(\tilde{B}) >0$ and therefore $ \ord_q(B) \geq -n+1$.
\end{proof}

% Let $\ell$ denote the dispersion of $A$ at $q$.  recall that a necessary
% condition for $A$ to be (partially) desingularizable at $q$ is that $\ell >0$. If
% this necessary condition is fulfilled then by repeating, almost $\ell$ times,
% the process described in the above proof one can compute a polynomial
% transformation $T$ such that the quantity $\max(0,- \ord_q(T[A]_\T)$ is as small
% as possible and $\den(T[A]_\T) \mid \den(A)$. In particular, if $q$ is removable
% then $T[A]_\T$ is desingularized at $q$.  This leads to the following \bigskip

This proof motivates the following alternative desingularization algorithm. In
contrast to Algorithm~1, instead of shifting a singularity towards a zero of the
system, it performs the analogous task of moving a zero towards the singularity
until they cancel each other

\begin{figure}[h]
  \begin{center}
\begin{algo}
  {desingularize\_B($A,q$)}
  {$A$ with entries in $\F(\n)$ and a single, simple, irreducible pole $q\in\F[\n]$.}
  {$(T,T[A]_\T)$ s.t.\ $T[A]_\T$ is desingularized at $q$.}
\item $T \leftarrow I_{d}$
\item \begin{while}{$\operatorname{den}(A) = 0\ \m\ q$}
  \item $\ell \leftarrow \T-\operatorname{dispersion}(A,q)$
  \item \textbf{IF} ($\ell \leq 0$) \textbf{THEN RETURN} ($T, A$)
  \item $n\leftarrow \ord_q(A)$ ; $\tilde{A} \leftarrow q^nA$
  \item $k \leftarrow 0$ ; $M \leftarrow I_{d}$ ;
    $N \leftarrow \pi_q(\tilde{A})$
  \item \begin{while}{$M\cdot N \neq 0$ \textbf{AND} $k \leq \ell$}
    \item $M \leftarrow M \cdot N$; $k \leftarrow k+1$ ;
      $N \leftarrow \pi_q(\T^{-k}(\tilde{A}))$
    \end{while}
  \item $U \leftarrow$ as in the proof of Proposition~\ref{prop:factorial0}.
  \item $D\leftarrow \operatorname{diag}(\T^{k-1}(q)\id_s,\id_{d-s})$ with
    $s =d-\operatorname{rank}(N)$.
  \item $A\leftarrow \T (U\cdot D)^{-1}\cdot A\cdot (U\cdot D)$
  \item $T\leftarrow T\cdot U\cdot D$
  \end{while}
\item \textbf{RETURN} ($T, A$)
\end{algo}
\end{center}
\end{figure}

An implementation of Algorithm~1 and Algorithm~2 in the computer algebra system Sage~\cite{sage} can be
obtained from
\begin{center}
\url{http://www.mjaroschek.com/systemdesing/}
\end{center}
\begin{remark}
 All systems that are desingularizable via Algorithm~1 are also desingularizable
 via Algorithm~2 and vice versa.
\end{remark}

\begin{example}
  For $A$ as in Example~\ref{ex:mainex} and $q=z-2$  we have
    \begin{alignat*}2
     & \tilde{A} = (\n-2)A =\begin{pmatrix}0 & \n-2 \\
        -2(\n +1) &3(\n - 1)\end{pmatrix},\\
  & M = \pi_q(\tilde{A}(\n)\tilde{A}(\n-1)\tilde{A}(\n-2))=
      \tilde{A}(2)\tilde{A}(1)\tilde{A}(0)=\smash[t]{\begin{pmatrix}0 & 0,\\
      -12 & 6\end{pmatrix}}\neq 0\\
  & N = \pi_q(\T^{-3}(\tilde{A})= \tilde{A}(-1)= \smash[t]{\begin{pmatrix}0 & -3 \\
      0 & -6\end{pmatrix}}, \\
  & \pi_q(\tilde{A}(\n)\tilde{A}(\n-1)\tilde{A}(\n-2)\tilde{A}(\n-3)) =
      \tilde{A}(2)\tilde{A}(1)\tilde{A}(0)\tilde{A}(-1)=0,
  \end{alignat*}
  so $k=3$.
  If we chose 
  \[ U= \smash{\begin{pmatrix}\frac12 & \frac12 \\
      0&1\end{pmatrix}},\] 
  then
  \[ \T(U)^{-1}AU=  U^{-1}AU= \frac{1}{(\n-2)}\begin{pmatrix}(\n+1) & 0 \\
      -(\n+1) &2(\n-2)\end{pmatrix}.\]
  We have $s=1$, so with
  $D = \operatorname{diag}(\T^2(q),1)= \operatorname{diag}(\n,1)$ we
  get
  \[B=\T(D)^{-1}(\T(U)^{-1}AU)D = \frac{1}{(\n-2)}\begin{pmatrix}\n & 0 \\
      -\n(\n+1) &2(\n-2)\end{pmatrix}.\]
  Note that, as expected, we have
  that \[\tilde{B}(2)\tilde{B}(1)\tilde{B}(0)=0.\]
  Here we can repeat the above process on $B$ to desingularize as much as
  possible the matrix $A$ at $q=z-2$.  In this particular example $q$ is removable by the transformation
  $T=U\cdot \operatorname{diag}(\n(\n-1)(\n-2),1)$. Indeed, one can see that
  \[ T[A]_\T=\T (T)^{-1} A T = \begin{pmatrix}1 & 0 \\
       -\n(\n^2-1) & 2\end{pmatrix},\]
has polynomial entries. The transformation $T$ is the same as in
Example~\ref{ex:mainex5} up to a right factor $\operatorname{diag}(\frac12,1)$.
\end{example}

\subsection{Rank Reduction}
Consider a system $[A]_{\T}$ and let $q$ be a $\T-$minimal factor of
$\den{(A)}$ with multiplicity $n\geq 1$, such that $[A]_{\T}$ the is not
partially desingularizable at $q$. This implies that there's no positive integer $k$ such
that relation~\eqref{eq:factorial0} holds.  As the quantity $n$ cannot be
reduced, it's natural to ask if it is possible to reduce the rank of the leading
matrix $\lc_q(A)$ by applying a polynomial transformation $T$ to $[A]_{\T}$. We
shall give a criterion for the existence of a polynomial transformation $T$ such
$ord_q(T[A]_\T) =ord_q(A)$ and $\ra(\lc_q(T[A]_\T)) < \ra(\lc_q(A))$
\begin{proposition} 
  \label{prop:factorial2}
  Let $q\in \F[z]$ be a $\T$-minimal pole of $[A]_\T$.  Let
  $\tilde{A}(z) = q^n A(z)$, so that $\ord_q(\tilde{A})=0$ and
  $\pi_q(\tilde{A})=\lc_q(A)$. %Suppose that $A$ is not partially
                               %desingularizable at $q$ (which means that $n$ is
                               %already minimal).
  Then a necessary and sufficient condition for the existence of a polynomial
  transformation $T$ such that $ord_q(T[A]) =ord_q(A)$ and
  $\ra(\lc_q(T[A])) < \ra(\lc_q(A))$ is that there exists a positive integer $k$
  such that
  \begin{equation}
    \label{eq:factorial2}
     \ra(\pi_q(\tilde{A}\T^{-1}(\tilde{A})\dots \T^{-k}(\tilde{A}))) < \ra(\lc_q(A)).
  \end{equation}
% TODO
%   Let $q\in \F[z]$ be a $\T$-minimal pole for the system $[A]_\T$.
%   Let $\tilde{A} = q^n A$, so that $\ord_q(\tilde{A})=0$ and
%   $\pi_q(\tilde{A})=\lc_q(A)$. 
%   Then a necessary and sufficient condition for the existence of a polynomial
%   transformation $T$ such that $\ra(\lc_q(T[A]_\T)) < \ra(\lc_q(A))$ and
%   $ord_p(T[A]_\T) \geq ord_p(A)$
%   % and $\ra(\lc_p(T[A]_\T)) \leq \ra(\lc_p(A))$
%   for all irreducible $p\in\F[\n]$,  is that there exists a positive integer $k$
%   such that
%   \begin{equation}
%     \label{eq:factorial2}
%      \ra(\pi_q(\tilde{A}\T^{-1}(\tilde{A})\dots \T^{-k}(\tilde{A}))) < \ra(\lc_q(A)).
%   \end{equation}
\end{proposition}

\begin{proof}
  {\noindent {\it Necessary condition}}: Suppose first that there exists a
  polynomial matrix $T$ with the desired properties and
  % such that $ord_q(T[A]) =ord_q(A)$ and
  % $\ra(\lc_q(T[A])) < \ra(\lc_q(A))$.
  let $B=T[A]_\T$. Similarly to the proof of
  Proposition~\ref{prop:factorial0}, one gets for all non-negative integers $k$:
  % Then for all non-negative
  % integers $k$ one has
  % \[\T(T) B\T^{-1}(B) \dots \T^{-k}(B) = A\T^{-1}(A) \dots \T^{-k}(A)
  %   \T^{-k}(T),\]
  % and hence
  \[\T(T) (q^nB)\T^{-1}(q^nB) \dots \T^{-k}(q^nB) = \tilde{A}\T^{-1}(\tilde{A})
    \dots \T^{-k}(\tilde{A}) \T^{-k}(T).\]
  Since $\ord_q(q^nB)=0=\ord_q(\tilde{A})$ and all the other factors in both
  sides of this equality have non-negative orders at $q$ we get that
  % \[ \ord_q(\T^{-j}(q^nB))= \ord_{\T^{j}(q)}(B) \geq \ord_{\T^{j}(q)}(A) \geq
  %   0, \; \hbox{for all} \, j \in \N^*,\]
  \begin{align*}
    & \pi_q(\T(T)) \pi_q((q^nB))\pi_q(\T^{-1}(q^nB) \dots \T^{-k}(q^nB)) = \\  & \qquad\qquad\qquad\qquad\pi_q(\tilde{A}\T^{-1}(\tilde{A}) \dots \T^{-k}(\tilde{A}))\pi_q( \T^{-k}(T)).
  \end{align*}

   By using the fact that the rank of a product of matrices is less or equal to
   the rank of each factor we get that the rank of the product in the right hand
   side of the previous equality is bounded by
   $\ra(\pi_q((q^nB)))= \ra(\lc_q(B))$ and hence
   \begin{align*}
& \ra(\pi_q(\tilde{A}\T^{-1}(\tilde{A}) \dots \T^{-k}(\tilde{A}))\pi_q(
                    \T^{-k}(T)) ) \leq {}\\
&\qquad\qquad\qquad\qquad\qquad\ra(\lc_q(B))< \ra(\lc_q(A)).
   \end{align*}
 Now let $k$ be the smallest positive integer such that the matrix
 $\pi_q(T(\n-k))$ is of full rank. Then
 \begin{align*}
 &\ra(\pi_q(\tilde{A}\T^{-1}(\tilde{A}) \dots \T^{-k}(\tilde{A})))=\\
 & \qquad\ra(\pi_q(\tilde{A}\T^{-1}(\tilde{A}) \dots \T^{-k}(\tilde{A}))\pi_q(
   \T^{-k}(T)) ) < \ra(\lc_q(A)).
 \end{align*}

 {\noindent {\it Sufficient condition}}: Let $r=\ra(\lc(A))$ and let $k$ be
 minimal so that~\eqref{eq:factorial2} holds. Put
 \[M:=\pi_q(\tilde{A}\T^{-1}(\tilde{A}) \cdots \T^{-k+1}(\tilde{A})) \quad
   \text{and} \quad N:= \pi_q(\T^{-k}(\tilde{A})).\]
 By definition of $k$, the matrix $M$ is nonzero, has the same rank $r$ as
 $\lc_q(A)$ and we have the strict inequality
 \[ \ra(M \cdot N) < r = \ra(M).\]
 This implies in particular that $\ra{(N)} < d=\dim(A)$. Let $s:=d-\ra{(N)}$. As in
 the proof of Proposition~\ref{prop:factorial1}, we can assume that $N$ has the following form:
 % and
 % construct a matrix
 %  %It follows that \[0 < \ra(M) \leq s:=d-\ra{(N)} < d.\]
 % $P \in \operatorname{GL}_d(\F[z]/\langle q\rangle)$ such that $P \cdot N$
 %  has its last $(d-s)$ rows linearly independent over $\F[z]/\langle q\rangle$
 %  while its $s$ first rows are zero.  Consider the matrix $U= \T^{k-1}(P^{-1})$
 %  as an element of $\operatorname{Mat}_d(\F[\n])$ then by applying the
 %  unimodular polynomial transformation $Y(z)=U(z)X(z)$, we can assume that the
 %  matrix $N:=\pi_q(\T^{-k}(\tilde{A}))$ has the following form:
  \[N=\begin{pmatrix} \zm_s&\zm_{s,d-s} \\{N_{2,1}}&N_{2,2}
             \end{pmatrix},\]
  where $N_{2,1}$ and $N_{2,2}$ are matrices with entries in
  $\F[z]/\langle q\rangle$ of size $(d-s) \times s$ and ${(d-s)}\times (d-s)$
  respectively, so that the last $d-s$ rows of $N$ are linearly independent over
  $\F[z]/\langle q\rangle$.  Let
  $M= (M_{i,j})_{1\leq i,j \leq 2}$ be partitioned in
  four blocks as $N$. Then we have
  \[ M\cdot N = \begin{pmatrix} {M_{1,2}}\\M_{2,2}
             \end{pmatrix} \cdot \begin{pmatrix} N_{2,1} & N_{2,2}
             \end{pmatrix}.\]
   As the matrix $ ({N_{2,1}} \;\; N_{2,2})$ is of full rank, we get that
    \[\ra{\begin{pmatrix} {M_{1,2}}\\M_{2,2}
             \end{pmatrix}} = \ra{(M\cdot N)} < r.\]
 Let $\tilde{A} = (\tilde{A}_{i,j})_{1\leq i,j \leq 2}$ be partitioned in
four blocks as $N$.  Then $\pi_q(\T^{-k}(\tilde{A}_{1,j})) = 0$ for $j = 1, 2$.
% In other words, the $s$ first rows of $\tilde{A}(z)$ are divisible by
% $\T^{k}(q)$.
Using the substitution $Y=DX$ where $ D =
\operatorname{diag}(\T^{k-1}(q)\id_s, \id_{d-s})$, we get a system
$[B]_\T$ of the form~\eqref{eq:systemb}.
% \begin{align*}
%    B := \T(D)^{-1}AD={} & q^{-n}\left(
%       \begin{array}{cc}
%         {\frac{\T^{k-1}(q)\tilde{A}_{1,1}}{\T^{k}(q)}}&
%         {\frac{\tilde{A}_{1,2}}{\T^{k}(q)}}\\{\T^{k-1}(q)\tilde{A}_{2,1}}
%         &{\tilde{A}_{2,2}}
%       \end{array} \right) =  \\
%     & q^{-n}\left(
%       \begin{array}{cc}
%         \T^{k-1}(q)\tilde{A'}_{1,1} & \tilde{A'}_{1,2}\\
%         \T^{k-1}(q)\tilde{A}_{2,1} & \tilde{A}_{2,2}
%       \end{array} \right)\raisebox{-0.2em}{,}
% \end{align*}
%   for some matrices $\tilde{A'}_{1,1}$, $\tilde{A'}_{1,2}$ with entries in
%   $\SF_q$.
  with $\den(B)\mid \den(A)$ and $\ord_q(B) \geq
  \ord_q(A)$.  Note that 
  \begin{align*}
    \pi_q(q^nB)= {}&
    \begin{pmatrix}
            \pi_q(\frac{1}{\T^{k}(q)})\id_s&\zm_{s,d-s}
            \\{\zm_{d-s,s}}&\id_{d-s} \end{pmatrix} \cdot
    \begin{pmatrix}
             \pi_q (\tilde{A}_{1,1} )&  \pi_q(\tilde{A}_{1,2})\\
             \pi_q(\tilde{A}_{2,1})  & \pi_q(\tilde{A}_{2,2})
           \end{pmatrix} \cdot\\
    & \begin{pmatrix}
            \pi_q({\T^{k-1}(q)})\id_s & \zm_{s,d-s} \\
            {\zm_{d-s,s}} & \id_{d-s}
          \end{pmatrix}\raisebox{-0.2em}{.}
  \end{align*}
      It follows that if $k \geq 2$, then $\ra(\pi_q(q^nB)) =\ra(\lc_q{(A)})$, but we will
      prove that the factorial relation~\eqref{eq:factorial2} holds for
      $\tilde{B}:=q^nB$ with $k-1$ instead of $k$. As in the proof of
      Proposition~\ref{prop:factorial1}, we have that% For this we remark first that
  % \[\T(D) \tilde{B}\T^{-1}(\tilde{B}) \cdots \T^{-k+1}(\tilde{B}) =
  %   \tilde{A}\T^{-1}(\tilde{A}) \cdots \T^{-k+1}(\tilde{A}) \T^{-k+1}(D).\]
  % It then follows that
  % \[\pi_q(\T(D)) \pi_q(\tilde{B}\T^{-1}(\tilde{B}) \cdots \T^{-k+1}(\tilde{B}))=
  %   M\cdot \pi_q(\T^{-k+1}(D)).\] We have that
  \[ \pi_q(\T^{-k+1}(D)) = \pi_q(\operatorname{diag}(q\id_s, \id_{d-s})) =
    \operatorname{diag}(\zm_s, \id_{d-s}),\]
  hence 
  \[ M\cdot \pi_q(\T^{-k+1}(D)) = \begin{pmatrix} \zm_s&{M_{1,2}}\\\zm_{d-s}&M_{2,2}
             \end{pmatrix},\]
   whose rank is less than $r$.  Now
  $\pi_q(\T(D)) = \operatorname{diag}(\pi_q(\T^k(q))\id_s, \id_{d-s})$ is
  invertible (since $q$ and $\T^k(q)$ are co-prime), it then follows that
  \begin{align*}
  & \ra(\pi_q(\tilde{B}\T^{-1}(\tilde{B}) \cdots \T^{-k+1}(\tilde{B}))) = \\
& \qquad\qquad\ra{(M\cdot \pi_q(\T^{-k+1}(D))}< r = \ra{(\lc_q{(B)})}.
  \end{align*}

  If $k-1$ is still positive then we can repeat this process on the matrix $B$
  and the polynomial~$q$ until we arrive at $k=1$.  Then we have that
  \[ \pi_q(q^nB)= \begin{pmatrix}
        \pi_q(\frac{1}{\T(q)})\id_s&\zm_{s,d-s}
        \\{\zm_{d-s,s}}&\id_{d-s} \end{pmatrix}\cdot\begin{pmatrix}
        \zm_s&{M_{1,2}}\\\zm_{d-s}&M_{2,2}
             \end{pmatrix},\]
  whose rank is less than $r$.
\end{proof}

The proof of Proposition~\ref{prop:factorial2} suggests that Algorithm~2 can be
easily adapted to minimize the rank of the leading matrix of a $\T$-minimal
pole. In particular, a $T$ can be computed such that
$\ord_p(T[A]_\T)\geq \ord_p(A)$ for $p\in\F[\n]$. It is to note that rank reduction for a pole in
$A$ via Algorithm~2 comes at the potential cost of an increase in order of a
pole of $A^*$, as the next example shows.

\begin{example}
\label{ex:orderincrease}
Consider the system with 
\[A=\begin{pmatrix}\n(\n+1) & 0 & 0\\
0 & \frac{\n+1}{\n} & 0 \\
0 & 0 & \frac{1}{\n}\end{pmatrix},\qquad A^*=
\begin{pmatrix}
\frac{1}{\n(\n-1)} & 0&0\\
0 &  \frac{\n-1}{\n} & 0 \\
0 & 0 & \n-1
\end{pmatrix}.\]
\noindent We have $\rank(\lc_\n(A)) =2$ and $\ord_{z-1}(A^*)=-1$, and computing a rank reducing
transformation for $[A]_\T$ via Algorithm~2 gives $T=\operatorname{diag}(\n,\n,1)$,
which results in
\[T[A]_\T=\begin{pmatrix}\n^2 &0&0\\
0 & 1 & 0\\
0 & 0 & \frac{1}{\n}\end{pmatrix},\qquad T[A]_\T^*=
\begin{pmatrix}
\frac{1}{(\n-1)^2} & 0 & 0\\
0 &  1 & 1\\
0 & 0 & \n-1
\end{pmatrix},\]
with $\rank(\lc_\n(T[A]_\T))= 1$ and $\ord_{z-1}(A^*)=-2$. We note that we merely
shifted an already present pole in $A^*$ to the right, as opposed to adding a new 
factor to the system.
% This shift, however, cannot be undone by a polynomial
% transformation, but would require a transformation with rational function
% entries: $\operatorname{diag}(1/\n,1,1)$.
\end{example}

\section{Apparent Singularities}
\label{sec:appsing}

In this section we establish the connection between the analytical notion of
apparent and removable singularities of meromorphic solutions and the algebraic
concept desingularization of difference systems. The key observation is the fact
that the factorial relation~\eqref{eq:factorial0} provides a sufficient
condition for a singularity to be removable..

\begin{proposition}
  \label{prop:factorial}
  Let $\p\in\mathcal{P}_r(A)$ be a pole of $A$ of order $\nu\geq 1$ such that
  $\p-j \notin P_r (A)$ for all positive integers $j$. Let
  $\tilde{A} = (\n -\p)^\nu A$, so that $\tilde{A}(\p) \neq 0$. If $\p$ is a
  removable r-singularity of $[A]_\T$, then there exists a positive integer $k$
  such that
  \begin{equation*}
    \tilde{A}(\p) A(\p-1) \cdots A(\p-k) = 0.
  \end{equation*}
  In particular, the matrix $A(\p -j)$ is singular for some non-negative integer
  $j$.
\end{proposition}

\begin{proof}
  Using a result due to Ramis~\cite{ramis,immink,barkatou}, one can easily prove
  that for any complex number $\eta$ with $-\operatorname{Re} \eta$ large enough,
  there exist a meromorphic fundamental matrix solution $F(\n)$ which is
  holomorphic for $-\operatorname{Re} z$ large enough and satisfies
  $F(\eta) = \id_d$.  Choose a positive integer $k$ such that
  $-\operatorname{Re}(\p -k)$ is large enough and take a fundamental matrix
  solution $F(\n)$ as above with $F(\p - k) = \id_d$. Then one can write
  \[F(\n + 1) = A(\n)A(\n - 1)A(\n - 2)\cdots A(\n - k)F(\n - k),\]
  and hence 
  \[(\n - \p)^\nu F (\n + 1) = (\n -\p)^\nu A(\n)A(\n - 1)A(\n - 2)\cdots A(\n -
    k)F (\n - k).\] Taking the limit as $\n$ goes to $\p$, we get that
  \[0 = \tilde{A}(\p)A(\p - 1)A(\p - 2)\cdots A(\p -k).\qedhere\]
\end{proof}

\begin{corollary}
  Let $\p \in P_r(A)$ such that there is a $\T$-minimal $q$ with $q(\p)=0$.  If $\p$ is a removable
  singularity of $[A]_\T$, then $[A]_\T$ is desingularizable at $q$.
\end{corollary}

\begin{proof}
  Let $n:= -\ord_q(A)$. We can apply Proposition~\ref{prop:factorial1} to reduce
  the multiplicity of $q$ in $\den(A)$ from $n$ to $n-1$. If $n>1$, $q$ is still
  $\T$-minimal and $\p$ still removable, and we can repeat the process until
  $[A]_\T$ is desingularized at $q$.
\end{proof}

\section{Conclusion and Future Work}
\label{sec:conclusion}
\balance In this paper we presented two algorithms to desingularize linear first
order difference systems and we explored the notions of apparent and removable
singularities. These topics have already been studied in the context of
difference operators, where usually the solution space of a given operator is
increased as a side effect of the desingularization process. An interesting
starting point for further research is to investigate the relation of
desingularization on a system level and on an operator level in regard to this
extension of the solution space.

Concerning pseudo linear systems, we will continue our work in several
directions. We aim to establish a clear connection between removable
singularities of a system $[A]_\T$ and the removable singularities of
$[A^*]_\IT$, as well as the role of removable singularities for extending
numerical sequences. Furthermore, studying desingularization at non-$\T$-minimal
poles is a promising approach for identifying poles that only appear in some
components of fundamental solutions. We are currently also investigating how to
characterize poles in solutions that do not propagate to infinitely many
congruent points via gauge transformations.

Regarding complexity, it would be desirable to conduct a thorough complexity
analysis of desingularization algorithms. Finally, as was shown
in~\cite{chen:curve,improv}, removable singularities of operators can negatively
impact the running time of some algorithms and it is interesting to investigate
whether similar effects occur for linear systems.
\bibliographystyle{plain}
\bibliography{main}
\end{document}